\RequirePackage{fix-cm}
\documentclass{svjour3}                     
\smartqed  
\usepackage{graphicx}
\usepackage{mathptmx}      
\usepackage{amsmath}
\usepackage{amssymb}
\usepackage{xspace}

\newcommand{\OPT}{\ensuremath{\operatorname{\textsc{Opt}}}\xspace}
\newcommand{\OPTALG}{\ensuremath{\operatorname{\textsc{Opt}_{\ALG}}}\xspace}
\newcommand{\ALG}{\ensuremath{\operatorname{\textsc{Alg}}}\xspace}
\newcommand{\GREEDY}{\ensuremath{\operatorname{\textsc{Greedy}}}\xspace}
\newcommand{\UFF}{\ensuremath{\operatorname{\textsc{Uff}}}\xspace}
\newcommand{\FAST}{\ensuremath{\operatorname{\textsc{Fast}}}\xspace}

\newcommand{\eps}{\ensuremath{\varepsilon}\xspace}

\journalname{Journal of Scheduling}
\begin{document}

\title{Online-Bounded Analysis\thanks{Supported in part by the Danish Council
for Independent Research, Natural Sciences, and the Villum Foundation.
A preliminary version of this paper appeared in the
{\it Eleventh International Computer Science Symposium in Russia},
Lecture Notes in Computer Science, vol.~9691,
Springer, 2016, pp.~131--145.
}}



\author{Joan Boyar
        \and Leah Epstein
        \and Lene M. Favrholdt
        \and Kim S. Larsen
        \and Asaf Levin
}

\authorrunning{Boyar, Epstein, Favrholdt, Larsen, Levin}

\institute{Joan Boyar \at
              Dept.\ of Mathematics and Computer Science \\
              University of Southern Denmark \\
              Odense, Denmark. \\
              \email{joan@imada.sdu.dk}
           \and
           Leah Epstein \at
              Dept.\ of Mathematics \\
              University of Haifa \\
              Haifa, Israel. \\
              \email{lea@math.haifa.ac.il}
           \and
           Lene M. Favrholdt \at
              Dept.\ of Mathematics and Computer Science \\
              University of Southern Denmark \\
              Odense, Denmark. \\
              \email{lenem@imada.sdu.dk}
           \and
           Kim S. Larsen \at
              Dept.\ of Mathematics and Computer Science \\
              University of Southern Denmark \\
              Odense, Denmark. \\
              \email{kslarsen@imada.sdu.dk}
           \and
           Asaf Levin \at
              Faculty of IE\&M \\
              The Technion \\
              Haifa, Israel. \\
              \email{levinas@ie.technion.ac.il}
}

\date{Received: date / Accepted: date}

\maketitle

\begin{abstract}
Though competitive analysis is often a very good tool for the
analysis of online algorithms, sometimes it does not give any
insight and sometimes it gives counter-intuitive results.
Much work has gone into exploring other performance measures,
in particular targeted at what seems to be the core problem
with competitive analysis: the comparison of the performance
of an online algorithm is made with respect to a too powerful adversary.
We consider a new approach to restricting the power
of the adversary, by requiring that when judging a given
online algorithm,
the optimal offline algorithm must perform at least as well as the online
algorithm, not just on the entire final request sequence, but also
on any prefix of that sequence.
This is limiting the adversary's usual advantage
of being able to exploit that it knows the sequence is continuing
beyond the current request.
Through a collection of online problems, including
machine scheduling, bin packing, dual bin packing, and seat reservation,
we investigate the significance of this particular offline advantage.
\end{abstract}

\maketitle

\section{Introduction}
\label{sec:intro} An \emph{online problem} is an optimization
problem where requests from a sequence~$I$ are given one at a
time, and for each request an irrevocable decision must be made
for it before the next request is revealed. For a minimization
problem, the goal is to minimize some cost function, and if \ALG
is an online algorithm, we let $\ALG(I)$ denote this cost on the
request sequence~$I$. Similarly, for a maximization problem, the
goal is to maximize some value function (also known as profit),
and in this case, $\ALG(I)$ is the profit of an online algorithm
\ALG.

\subsection{Performance Measures}
\label{sec:measures}
Competitive analysis~\cite{ST85,KMRS88} is the most common tool
for comparing online algorithms.
For a minimization problem,
an online algorithm is \emph{$c$-competitive}
if there exists a constant~$\alpha$ such that for all input
sequences~$I$, $\ALG(I)\leq c\OPT(I)+\alpha$.
Here, $\OPT$ denotes an optimal offline algorithm.
As usual, the term ``offline'' is just used for emphasis,
since most algorithms we discuss are online.
The (asymptotic) \emph{competitive ratio} of \ALG
is the infimum over all such~$c$.
Similarly, for a maximization problem, an online algorithm is
\emph{$c$-competitive} if there exists a constant~$\alpha$
such that for all input
sequences~$I$, $\ALG(I)\geq c\OPT(I)-\alpha$.
The (asymptotic) \emph{competitive ratio} of \ALG
is the supremum over all such~$c$.
In both cases,
if the inequality can be established using~$\alpha=0$, we refer to the
result as being \emph{strict}
(some authors use the terms \emph{absolute} or \emph{strong}).
Note that for maximization problems,
we use the convention of competitive ratios smaller than~$1$.

For many online problems, competitive analysis gives useful and
meaningful results. However, researchers also realized from the
very beginning that this is not always the case: Sometimes
competitive analysis does not give any insight and sometimes it
even gives counter-intuitive results, in that it points to the
worse of two algorithms as the better one (in the sense that the
common belief is that one of the two algorithms is worse, or even
that experimental studies provide clear evidence that this is the
case). A recent list of examples with references can be found
in~\cite[p.~289]{EKL13}. Much work has gone into exploring other
performance measures, in particular targeted at what seems to be
the core problem with competitive analysis: the comparison of the
performance of an online algorithm is made with respect to a too
powerful adversary.

Four main techniques for addressing this have been employed,
sometimes in combination. We discuss these ideas below.
No chronological order is implied by the order the techniques are
presented in.
First, one could completely eliminate the optimal offline algorithm by
comparing algorithms to each other directly. Measures taking this approach
include max/max analysis~\cite{BDB94},
relative worst order analysis~\cite{BF07},
bijective and average analysis~\cite{ADLO07}, and
relative interval analysis~\cite{DLM09}.

Second, one could limit the resources of the optimal offline algorithm,
or correspondingly increase the resources of the online algorithm,
as is done in extra resource analysis~\cite{KP00,ST85}.
Thus, the offline algorithm's knowledge of the future is counter-acted
by requiring that it solves a harder version of the problem than the online
algorithm.
Alternatively, the online algorithm could be given limited knowledge
of the future in terms of some form of look-ahead, as  has been done for
paging. In those set-ups, one assumes that the online
algorithm can see a fixed number $\ell$ of future requests, though it
varies whether it is simply the next $\ell$ requests, or, for instance,
the next $\ell$ expensive requests~\cite{Y91},
the next $\ell$ new requests~\cite{B98},
or the next $\ell$ distinct requests~\cite{A97}.

Third, one could limit the adversary's control over exactly which
sequence is being used to give the bound by grouping sequences
and/or considering the expected value over some set as has been done
with the
statistical adversary~\cite{R91},
diffuse adversary~\cite{KPn00},
random order analysis~\cite{K96},
worst order analysis~\cite{BF07},
Markov model~\cite{KPR00}, and
distributional adversary~\cite{Giannakopoulos15}.

Finally, one could limit the adversary's choice of sequences it is allowed
to present to the online algorithm.
An early approach to this, which at the same time addressed
issues of locality of reference, was the access graph model~\cite{BIRS95},
where a graph defines which requests are allowed to follow each other.
Another locality of reference approach was taken in~\cite{AFG02},
limiting the maximum number of different requests allowed within some
fixed-sized sliding window.
Both of these models were targeted at the paging problem, and the
techniques are not meant to be generally applicable to online algorithm
analysis.
A resource-based approach is taken in~\cite{BL99}, where only sequences
that could be fully accommodated given some resource are considered,
eliminating some pathological worst-case sequences.
A generalization of this, where the competitive ratio is found in
the limit, appears in~\cite{BLN01,BFLN03}.
All of these approaches are aimed at removing pathological sequences
from consideration such that the worst-case (or expected case)
behavior is taken over a smaller and more realistic set of sequences,
thereby obtaining results aligning better with observed behavior
in practice. A similar concept for scheduling problems is the ``known-\OPT''
model, where the cost of an optimal offline solution is known
in advance~\cite{AR01}.
Finally, loose competitive analysis~\cite{Y94} allows for
a set of sequences, asymptotically smaller than the whole
infinite set of input sequences, to be disregarded,
while the remaining sequences should either be $c$-competitive
or have small cost. In this way, infrequent pathological
as well as unimportant (due to low cost) sequences can be eliminated.

\subsection{Online-Bounded Analysis}
\label{subsec:bounded}
Much work has been done in all of the four categories mentioned above.
In this paper, we consider a new approach to restricting the power
of the adversary that does not really fit into any
of the known categories.
Given an online algorithm, we require that
the optimal offline algorithm perform at least as well as the online
algorithm, not just on the entire final request sequence, but also
on any prefix of that sequence.
In essence, this is limiting the adversary's usual advantage
of being able to exploit that it knows the sequence is continuing
beyond the current request, without completely eliminating this advantage.
Since the core of the problem of the adversary's strength is its
knowledge of the future, it seems natural to try to limit that
advantage directly.

This new measure is generally applicable to online problems, since
it is only based on the objective function.
Comparing with other measures,
it is a new element that the behavioral restriction imposed on
the optimal offline algorithm is determined by the online algorithm,
which is the reason we name this technique \emph{online-bounded analysis}.
It is adaptive in the sense that online algorithms attempting non-optimal
behavior face increasingly harder conditions from the adversary the farther
the online algorithm goes in the direction of non-optimality (on prefixes).
The measure judges greediness more positively than does competitive analysis,
since making greedy choices limits the adversary's options more, so the
focus shifts towards the quality of a range of greedy or near-greedy decisions.

Behavioral restrictions on the optimal offline algorithm
have been seen before, as
in~\cite{CLLLT11}, where it is used as a tool to arrive at the final
result. Here they first show a $O(1)$-competitive result against
an offline algorithm restricted to, among other things,
using shortest remaining processing time for job selection.
Later they show that this gives rise to a schedule at most three
times as bad as for an unrestricted offline algorithm.
Thus, the end goal is the usual competitive ratio, and the
restriction employed in the process is problem specific.

\subsubsection{The Definition}
\label{subsubsec:new}
We start by giving the definitions for a minimization problem.

If $I$ is an input sequence for some minimization problem
and $A$ is a deterministic online algorithm for this problem,
we let $A(I)$ denote the objective function value returned by~$A$
on the input sequence~$I$.

We let $\OPT_A$ denote an offline algorithm which is optimal
\emph{under the restriction} that it is \emph{not}
allowed to be worse than~$A$ on any prefix of the input sequence~$I$
being considered.

Thus, for any sequence~$I$, for which we want to determine
$\OPT_A(I)$, it must hold for all prefixes~$I'$ of~$I$ that
$\OPT_A(I') \leq A(I')$. Additionally, no offline algorithm with
that property is strictly better than $\OPT_A$ on $I$. If these
conditions are fulfilled, we say that $\OPT_A$ is an
\emph{online-bounded optimal solution} (for~$A$).

If for some constants, $b$ and~$c$,
it holds for all sequences~$I$ that $A(I) \leq c \OPT_A(I)+b$,
then we say that $A$ has an \emph{online-bounded ratio} of at most~$c$.
The online-bounded ratio of~$A$ is the infimum over all such~$c$.

For a maximization problem, the requirement is instead that
$\OPT_A(I') \geq A(I')$, and
if for some constants, $b$ and~$c$,
it holds for all sequences~$I$ that $A(I) \geq c \OPT_A(I)-b$,
then we say that $A$ has an \emph{online-bounded ratio} of at least~$c$.
The online-bounded ratio of~$A$ is the supremum over all such~$c$.

For maximization problems, it varies whether
authors use ratios greater than or smaller than one.
Note that with our definitions, an online-bounded ratio for a
minimization problem is at least~$1$,
while this ratio for a maximization problem is at most~$1$.

Just as with the competitive ratio, one could also define a strict
variant of the online-bounded ratio with $b=0$.
For the scheduling problems considered in this paper, this would not
change the results. As has also been observed for competitive
analysis, for any constant $b$, the job sizes of a worst-case input can
be scaled such that allowing the additive constant $b$ makes no
difference.
Thus, for simplicity, in Sections~\ref{sec:scheduling}
and~\ref{sec:santa},
we assume $b=0$. When considering the seat reservation
problem, the total number of reservations accepted is bounded by a constant,
so we also assume that $b=0$ in Section~\ref{sec:unit}.

\subsection{Results}
\label{subsec:results} Through a collection of online problems,
including machine scheduling, bin packing, dual bin packing, and
seat reservation, we investigate the workings of online-bounded
analysis. The large collection of measures that have been defined
indicates that there is no universal measure which is the best
choice for all problems. With our approach, we try to learn more
about the nature of online problems, greediness, and robustness.
As a first approach, we study our new idea in the simplest
possible setting, and leave it for future work to investigate if
our approach works best in isolation or in combination with ideas
from other measures.

First, we observe that some results from competitive analysis carry over.
Then we note that some problem characteristics imply that a greedy algorithm
is optimal.

For machine scheduling, we obtain the following results. For
minimizing make\-span on $m\geq 2$ identical machines, we get an
online-bounded ratio of $2-\frac{1}{m-1}$ for \GREEDY. Though this
is smaller than the competitive ratio of
$2-\frac{1}{m}$~\cite{G66}, it is a comparable result,
demonstrating that non-greedy behavior is not the key to the
adversary performing better by a factor close to $2$ for
large~$m$. Two machines are called uniformly related if there
exists a fixed factor~$s$ such that one machine is $s$ times
faster than the other, that is, the two machines have speeds, and
$s$ is the ratio between those speeds, also called the speed ratio.
For two uniformly related machines, we prove that \GREEDY has
online-bounded ratio~$1$. This is consistent with competitive
ratio results, where \GREEDY has been proven
optimal~\cite{ENSSW01,CS80} with competitive ratio $1+\min\{\frac
1s, \frac s{s+1} \}$, if the speed ratio is $s$. For the case
where the faster machine is at least~$\phi$ (the golden ratio)
times faster than the slower machine, competitive analysis finds
that \GREEDY and \FAST, the algorithm that only uses the faster
machine, are equally good. Using relative worst order analysis,
\GREEDY is deemed the better algorithm~\cite{EFK06}, which seems
reasonable since \GREEDY is never worse on any sequence than
\FAST, and sometimes better. We also obtain this positive
distinction, establishing the online-bounded ratios~$1$ and
$\frac{s+1}{s}$
for \GREEDY and \FAST,
respectively.

For the Santa Claus machine scheduling problem~\cite{BS06},
we prove that \GREEDY
is optimal for identical machines with respect to the online-bounded ratio.
For two related machines with speed ratio~$s$, we present an algorithm with
an online-bounded ratio better than~$\frac1s$
and show that no online algorithm has a higher online-bounded ratio.
For this problem, it is known that the best possible competitive ratio
for identical machines is~$\frac 1m$, and the best possible competitive ratio
for two related machines is~$\frac1{s+1}$~\cite{Woe97,AE98,Eps05}.

For classic bin packing, we show that any Any-Fit algorithm has an
 online-bound\-ed ratio of at least~$\frac{3}{2}$.
We observe that for bin covering, the best online-bounded ratio
is equal to the best competitive ratio~\cite{CT88}.
For these problems, asymptotic measures are used.

We show a connection between results concerning the competitive
ratio on accommodating sequences
(that is, sequences where \OPT packs all items)
and the online-bounded ratio.
For dual bin packing (namely, the multiple knapsack problem with equal
capacity knapsacks and unit value items),
we show that the online-bounded ratio is the
same as the competitive ratio on accommodating sequences for a large
class of algorithms including First-Fit, Best-Fit, and Worst-Fit.
It then follows from results in~\cite{BFLN03} that any algorithm in
this class has an online-bounded ratio of at least~$\frac12$.
Furthermore, the online-bounded ratio of First-Fit and Best-Fit is~$\frac58$,
and that of Worst-Fit is~$\frac12$.
We also note that, for any dual bin packing algorithm, an upper bound
on the competitive ratio on accommodating sequences is also an upper
bound on the online-bounded ratio.
Using a result from~\cite{BFLN03}, this implies that any
algorithm has an online-bounded ratio of at most~$\frac67$.

For seat reservation, we have preliminary results, and conjecture
that results are similar to machine scheduling for identical machines,
in that ratios similar to but slightly better than those obtained
using competitive analysis can be established.

We found that the new measure sometimes leads to the same results
as the standard competitive ratio, and in some cases it leads to
an online-bounded ratio of~$1$.
However, there are problem variants for which we obtain an
intermediate value, which confirms the relevance of our approach.

\section{Online-Bounded Analysis}
\label{sec:bounded}
Before considering concrete problems, we discuss some generic properties.

\subsection{Measure Properties}
\label{model}
The online-bounded ratio of an algorithm is never further away from~$1$
than the competitive ratio, since the online algorithm's performance
is being compared to a (possibly) restricted, optimal algorithm.

Since algorithms are compared with different optimal algorithms,
one might be concerned that two algorithms, $A$ and~$O$, could
have online-bounded ratio~$1$, and yet one algorithm could do better on
some sequences than the other.
However, if both algorithms have online-bounded ratio~$1$, there is no
point where one
algorithm makes a decision which changes the objective value
more than the other does, since the adversary could end the
sequence there and the one algorithm with the worst objective value
would not have online-bounded ratio~$1$.
Thus, both algorithms have the same objective function value at
all points, so they always compete against the same adversary.
Thus, if algorithm $A$ performs better than algorithm $O$ on any
input sequence, then
algorithm~$O$ does not have online-bounded ratio~$1$.

For some problems, such as paging, $\OPT_A$ is the same as \OPT
under competitive analysis for all algorithms~$A$, because \OPT's
behavior on any sequence is also optimal on any prefix of that
sequence. Thus, the
competitive analysis results for paging and similar problems
also hold with this measure, giving the same online-bounded ratio
as competitive ratio.

\subsection{$\GREEDY$ is Sometimes Optimal}
\label{subsec:greedy} It is sometimes the case that there is one
natural greedy algorithm that always has a unique greedy choice in
each step. In such situations, the greedy algorithm is optimal
with respect to this measure, having online-bounded ratio~$1$. For
example, consider weighted matching in a graph where the edges
arrive in an online fashion (the edge-arrival model) and the
algorithm in each step decides if the current edge is added to the
matching or discarded. Here, the greedy algorithm, denoted by
$\GREEDY$, adds the current edge if adding the edge will keep the
solution feasible (that is, its two end-vertices are still exposed
by the matching that the algorithm created so far) and the weight
of the edge is strictly positive. Note that indeed the
online-bounded ratio of $\GREEDY$ is~$1$, as the solution
constructed by $\OPT_{\GREEDY}$ must coincide with the solution
created by $\GREEDY$. The last claim follows by a trivial
induction on the number of edges considered so far by both
$\GREEDY$ and $\OPT_{\GREEDY}$. If $\GREEDY$ adds the current
edge, then by the definition of $\OPT_{\GREEDY}$, we conclude that
$\OPT_{\GREEDY}$ adds the current edge. If $\GREEDY$ discards the
current edge because at least one of its end-vertices is matched,
then $\OPT_{\GREEDY}$ cannot add the current edge either (using
the induction assumption). Last, if $\GREEDY$ discards the current
edge since its weight is non-positive, then we can remove the edge
from the bounded optimal solution, $\OPT_{\GREEDY}$, if it was
added (removing it from $\OPT_{\GREEDY}$ will not affect the
future behavior of $\OPT_{\GREEDY}$ since $\OPT_{\GREEDY}$ must
accept an edge whenever $\GREEDY$ does). Similar proofs hold in
other cases when there is a unique greedy choice for $\OPT$ in
each step. Note that for the weighted matching problem where
vertices arrive in an online fashion and when a vertex arrives the
edge set connecting this vertex to earlier vertices is revealed
with their weights (the vertex-arrival model), the standard
negative result (the value of the negative result is smaller than
the ratio of the smallest strictly positive weight in the graph to
the largest) for weighted matching holds as can be seen in the
following construction. The first three vertices arrive in the
order~$1,2,3$ and when vertex~$3$ arrives, two edges
$\{1,3\},\{2,3\}$ are revealed each of which has weight of~$1$
(vertices~$1$ and~$2$ are not connected). At this point, an online
algorithm with a strictly positive online-bounded ratio must add
one of these edges to the matching. Then, either~$1$ or~$2$ are
matched in the current solution, and in the last step, vertex~$4$
arrives with an edge of weight~$M$ connecting~$4$ to the vertex
among~$1$ and~$2$ that was matched by the algorithm. Observe that
when vertex~$3$ arrives, the algorithm adds an edge to the
matching while the bounded optimal solution can add the other
edge, and this will allow the bounded optimal solution to add the
last edge as well.

The argument for the optimality of $\GREEDY$ for the weighted matching
problem in the edge-arrival model clearly holds if all weights are~$1$
also.
This unweighted matching problem in the edge-arrival model is
an example of a maximization problem in the online complexity class
Asymmetric Online Covering (AOC)~\cite{BFKM15}:
\begin{definition} \label{sgeasydef}\label{wdef}
An online accept-reject problem 
is in \emph{Asymmetric Online Covering (AOC)} if,
for the set~$Y$ of requests accepted, the following holds:

For minimization (maximization) problems,
the objective value of~$Y$ is $|Y|$ if $Y$ is
feasible and $\infty$ ($-\infty)$ otherwise, and any superset (subset)
of a feasible solution is feasible.
\end{definition}

For all maximization problems in the class AOC, there is an obvious greedy
algorithm, $\GREEDY$, which accepts a request whenever acceptance
maintains feasibility.
The argument above showing that the online-bounded ratio of
$\GREEDY$ is~$1$ for the weighted matching problem in the edge-arrival
model generalizes to all maximization problems in AOC.

\begin{theorem}
For any maximization problem in AOC, the online-bounded ratio of
$\GREEDY$ is~$1$. Thus, $\GREEDY$ is optimal according to
online-bounded analysis for online independent set in the vertex-arrival
model, unweighted matching in the edge-arrival model, and online
disjoint path allocation where requests are paths.
\end{theorem}

Note that this does not hold for all minimization problems in AOC.
For example, cycle finding in the vertex-arrival model, the problem
of accepting as few vertices as possible, but accepting enough so
that there is a cycle in the induced subgraph accepted, is AOC-Complete.
However, consider the first vertex requested in a graph with only
one cycle. $\GREEDY$ is forced to accept it, since the vertex could
be part of the unique
cycle, but $\OPT_{\GREEDY}$ will reject the vertex if it is not in that cycle.

However, there are online-bounded optimal greedy algorithms for
minimization problems in AOC, such as vertex cover, which
are \emph{complements} of maximization problems in AOC (independent set in
the case of vertex cover). By complement, we mean that set~$S$ is a
maximal feasible set in the maximization problem if and only if
the requests not in~$S$ are a feasible solution for the minimization
problem. The greedy algorithm in the case of these minimization
problems would be the algorithm that
accepts exactly those requests that $\GREEDY$ for the complementary
maximization problem rejects.

\section{Machine Scheduling: Makespan}
\label{sec:scheduling}
We study the load balancing problem of
minimizing makespan for online job scheduling without preemption.
We first consider $m$ identical machines, and analyze the classic greedy
algorithm (also known as list scheduling). At any point, \GREEDY
schedules the next job on a least loaded machine, where the load
of a machine is the sum of the sizes of all jobs assigned to it.
Since the machines are identical, ties can be
resolved arbitrarily without loss of generality. It is known that
the competitive ratio of \GREEDY is $2-\frac{1}{m}$~\cite{G66}.
With the more restricted optimal algorithm, we get a smaller value
of $2-\frac{1}{m-1}$ as the online-bounded ratio of \GREEDY. Any
algorithm assigns every job to run within a specific time window
of this job, and the completion time of a job is the ending point
of its window.

\begin{lemma}
\label{lemma-makespan-upper}
For the problem of minimizing makespan for online job scheduling
on $m$ identical machines,
\GREEDY has online-bounded ratio of at most~$2-\frac{1}{m-1}$.
\end{lemma}
\begin{proof}
Consider a sequence~$I$. Let $j$ be the first job in~$I$ that is completed
at the final makespan of $\GREEDY$, and assume that it has size~$w$.
Let $t$ and~$t'$ be the starting times of~$j$ in $\OPT_{\GREEDY}$ and $\GREEDY$,
respectively,
and let $\ell$ and~$\ell'$ be the makespans of $\OPT_{\GREEDY}$ and $\GREEDY$,
respectively,
just before the arrival of~$j$.
Let $V$ denote the sum of the sizes of the jobs in~$I$ just before $j$ arrives.

We have the following inequalities:
\begin{itemize}
\item $\OPT_{\GREEDY}(I) \geq t+w$
\item $\OPT_{\GREEDY}(I) \geq \ell$
\end{itemize}
In addition, since, just before $j$ arrived,
the machine where $\OPT_{\GREEDY}$ placed
$j$ had load~$t$ and the other machines had load at most~$\ell$,
$V\leq t+(m-1)\ell$. Since $m-1\geq 1$, $V\leq (m-1)(t+\ell)$.

Because $\GREEDY$ placed~$j$ on its least loaded machines,
all machines had load at least~$t'$ before $j$ arrived. At least one
machine had load $\ell'$, so $V\geq (m-1)t'+\ell'$.
By the definition of online-bounded analysis,~$\ell \leq \ell'$.
Thus, $V\geq (m-1)t'+\ell$.
Combining the upper and lower bounds on~$V$ gives
$(m-1)t'\leq (m-1)t+(m-2)\ell$ and $t' \leq t+\frac{m-2}{m-1}\ell$.
We now bound $\GREEDY$'s makespan:
\[
\begin{array}{lcl}
\GREEDY(I) & = & t'+w
 =  (t'-t)+(t+w) \\
& \leq  & \left(\frac{m-2}{m-1}\right)\cdot\ell + \OPT_{\GREEDY}(I)
 \leq   \left(2-\frac{1}{m-1}\right)\OPT_{\GREEDY}(I)
\end{array}
\]
\mbox{}\qed
\end{proof}

\begin{lemma}
\label{lemma-makespan-lower}
For the problem of minimizing makespan for online job scheduling
on~$m$ identical machines,
\GREEDY has online-bounded ratio of at least~$2-\frac{1}{m-1}$.
\end{lemma}
\begin{proof}
The adversarial sequence~$I$ consists of one job of size~$m-1$, followed by
$(m-1)(m-2)$ jobs of size~$1$, and finally one job of size~$m-1$.
Clearly, $\GREEDY(I)=2m-3$.

Since the makespan of \GREEDY after the first job is~$m-1$,
$\OPT_{\GREEDY}$ is allowed to schedule the $(m-1)(m-2)$ jobs
on $m-2$ machines,
all of which are different from the machine getting the first job,
until $m-1$ machines all have load~$m-1$.
This leaves one machine for the final job, and gives a final makespan of~$m-1$.
The online-bounded ratio becomes $\frac{2m-3}{m-1}=2-\frac{1}{m-1}$.
\mbox{}\qed
\end{proof}

By Lemmas~\ref{lemma-makespan-upper} and~\ref{lemma-makespan-lower},
we find the following.
\begin{theorem}\label{makespanthm}
For the problem of minimizing makespan for online job scheduling
on $m$ identical machines,
\GREEDY has online-bounded ratio~$2-\frac{1}{m-1}$.
\end{theorem}

Note that Theorem~\ref{makespanthm} establishes the existence
of an online algorithm, $\GREEDY$, for makespan minimization
on two identical machines with an online-bound\-ed ratio of~$1$.
Next, we generalize this last result to the case of
two uniformly related machines.
Note that for two uniformly related machines we can assume that
machine number~$1$ is strictly faster than machine number~$2$,
and the two speeds are $s>1$ and~$1$.
The load of a job assigned to a machine with speed $s'$ is the size of
the job divided by $s'$, and the load of a machine is the sum of
the loads of the jobs assigned to it.

We define \GREEDY as the algorithm that assigns the current job
to the machine such that adding the job there results in a solution
of a smaller makespan breaking ties in favor of assigning the job
to the slower machine (that is, to machine number~$2$).
If an algorithm breaks ties in favor of assigning the job
to the faster machine (let this algorithm be called $\GREEDY'$),
then its online-bounded ratio is strictly above~$1$,
as the following example implies. The first job has size~$s-1$
(and it is assigned to machine~$1$), and the second job has size~$1$
(and assigning it to any machine will result in the current makespan~$1$).
The first job must be assigned to machine~$1$ by~$\OPT_{\GREEDY'}$,
and it assigns the second job to the second machine.
A third job of size $s+1$ arrives.
This job is assigned to the first machine by~$\OPT_{\GREEDY'}$,
obtaining a makespan of~$2$.
$\GREEDY'$ will have a makespan of at least $\min\{2+1/s,s+1\}>2$ as~$s>1$.

\begin{theorem}\label{makespanrel}
For the problem of minimizing makespan for online job scheduling
on two uniformly related machines,
\GREEDY has online-bounded ratio~$1$.
\end{theorem}
\begin{proof}
Consider an input, and assume by contradiction that the makespan
of $\GREEDY$ exceeds that of $\OPT_{\GREEDY}$. Consider the last
time that the loads of the two machines of $\GREEDY$ are equal
(this time may be before any jobs arrive or later). Since the
makespan of $\OPT_{\GREEDY}$ at that time cannot be lower than the
makespan of $\GREEDY$, $\OPT_{\GREEDY}$ must have the same
makespan and its machines also have equal loads (it is possible
that the schedules are not identical). After this time, the
machines of $\GREEDY$ never have equal loads. If, starting at this
time and until the input ends, $\GREEDY$ and $\OPT_{\GREEDY}$
select the same machine for every job, then they will have the
same final makespan. Thus, there is a job that they assign to
different machines. Let $j$ be the first such job.  Before $j$ is
assigned, starting the last time that the machines had equal
loads, the two solutions have the same load for each of the two
machines as they received the same jobs.  As $\OPT_{\GREEDY}$
cannot obtain a larger makespan than ${\GREEDY}$, while
${\GREEDY}$ selects a machine that minimizes the makespan, it must
be the case that no matter which machine receives~$j$, the
resulting makespan will be the same. Thus, $\GREEDY$ assigns $j$
to the slower machine because of its tie breaking rule, while
$\OPT_{\GREEDY}$ assigns the job to the faster machine.  Since for
both solutions the loads of both of the two machines were lower
than the makespan achieved after $j$ is assigned, the machine that
achieves the makespan is unique (since there are only two
machines), and for each solution, it is the machine that received
the current job. If the makespan of \GREEDY does not increase in
any future step, then its final makespan cannot exceed that
of~$\OPT_{\GREEDY}$. Thus, assume that there is at least one such
future increase of the makespan of the solution constructed by
\GREEDY and consider the first such future step. Let $j'$ be the
job assigned at that step.

By the definition of $j'$, at the time when $j'$ arrives, the load of
the faster machine is no larger than the load of the slower machine.
Hence, by the definition of \GREEDY, 
the assignment of any job arriving later than
$j$ and up to and including~$j'$ (by \GREEDY) is to the faster
machine. 
Since $\OPT_{\GREEDY}$ cannot assign any
job that arrives after~$j$ but before~$j'$ such that its makespan
increases (since the makespan of \GREEDY does not increase),
$\OPT_{\GREEDY}$ assigns all these jobs to the slower machine. Let
$X_1$ and~$X_2$ be the total sizes of jobs that were assigned to the
two machines (in both solutions) before the arrival of~$j$, $p_j$
and~$p_{j'}$ the sizes of $j$ and~$j'$, and $Z$ the size of jobs
that arrived after~$j$ but before~$j'$. Recall that we have
\begin{equation}
\label{eq1}
\frac{X_1+p_j}s=X_2+p_j
\end{equation}
and this is the value of the makespan
(of both solutions) after $j$ was assigned.
Let $I'$ be the prefix of the input sequence ending with $j'$.
We find that the makespan of $\GREEDY$ after $j'$ is assigned is
$$\GREEDY(I') = \frac{X_1+Z+p_{j'}}s\,.$$
Since $\frac{X_1+Z}s < X_2+p_j=\frac{X_1+p_j}s$ while
$\frac{X_1+Z+p_{j'}}s >X_2+p_j=\frac{X_1+p_j}s$, we have
\begin{equation}
\label{eq2}
Z<p_j<Z+p_{j'}.
\end{equation}
By (\ref{eq2}) and $s>1$,
$X_1+Z+p_{j'} = X_1 + (1-s)(Z+p_{j'}) + s(Z+p_{j'}) < X_1+(1-s)p_j+s(Z+p_{j'})$.
Thus, $$\GREEDY(I') < \frac{X_1+p_j}{s} - p_j + Z + p_{j'}\,.$$
The makespan of $\OPT_{\GREEDY}$ after $j'$ is assigned is
$$\OPT_{\GREEDY}(I') \geq
      \min\left\{X_2+Z+p_{j'},\frac{X_1+p_j+p_{j'}}s\right\}\,.$$
By~(\ref{eq1}), $X_2+Z+p_{j'}=\frac{X_1+p_j}s-p_j+Z+p_{j'}$, and
by~(\ref{eq2}), $\frac{X_1+p_j+p_{j'}}{s} > \frac{X_1+Z+p_{j'}}s$.
Thus, when $j'$ arrives, $\OPT_{\GREEDY}$ cannot assign it without
increasing its makespan beyond the makespan of \GREEDY, contradicting
the definition of an online-bounded optimal solution.
\mbox{}\qed
\end{proof}

We now consider the algorithm \FAST that simply schedules all jobs on
the faster machine.
In contrast to \GREEDY, \FAST does not have an online-bounded ratio of~$1$.
This also contrasts with competitive analysis, since \FAST has an
optimal competitive ratio for~$s \geq \phi$,
where $\phi =\frac{1+\sqrt{5}}{2} \approx 1.618$.

\begin{theorem}\label{fastbound}
For two related machines with speed ratio~$s$, \FAST has an
online-bounded ratio of~$\frac{s+1}{s}$.
\end{theorem}
\begin{proof}
For the upper bound, consider any input sequence~$I$ and let $P$
denote the total size of the jobs in~$I$.
Then, $$\FAST(I)=\frac{P}{s} \text{ and } \OPT(I) \geq
\frac{P}{1+s}\,,$$
yielding a ratio of at most~$\frac{s+1}{s}$.

For the lower bound, consider the sequence $\langle s^2,s \rangle$.
Both \OPT and \FAST schedule the first job on the faster machine.
However, for the second job, \OPT will use the slower machine,
obtaining a makespan of~$s$.
Placing both jobs on the faster machine, \FAST ends up with a makespan
of~$s+1$.
\mbox{}\qed
\end{proof}

By Theorem~\ref{makespanthm},
the result of Theorem~\ref{makespanrel} cannot be extended to three
or more identical machines for $\GREEDY$.
We conclude this section by proving that such a generalization is
impossible, not only for \GREEDY, but for any deterministic online
algorithm.

\begin{theorem}\label{makespanlb}
Let $m\geq 3$. For the problem of minimizing makespan for
online job sched\-uling on $m$ identical machines,
any deterministic online algorithm~$A$ has online-bound\-ed ratio
of at least~$\frac{4}{3}$.
\end{theorem}
\begin{proof}
The input sequence starts with $m-2$ jobs of size~$3$ followed by two jobs
of size~$1$.  At this point, the makespan of the solution created by
$A$ is either~$3$ and in this case we continue, or at least~$4$
and in this case we stop the sequence.

If we decide to continue and the two jobs of size~$1$
are assigned to a common machine (thus $A$ has $m-2$ machines each
with load of~$3$ and another machine with load~$2$), the
sequence is augmented by two jobs, each of size
$2$, and thus the resulting makespan of~$A$ is at least
$4$.  Otherwise, if we decide to continue and the two
jobs of size~$1$ are assigned to different machines (and
thus the load of every machine in~$A$ is at least~$1$),
then the sequence is augmented by one job of size~$3$ and so the
resulting makespan of~$A$ is at least~$4$.  Note that in
all cases the makespan of~$\OPT_A$ is~$3$.
This holds since after the processing of the first job,
an optimal algorithm is allowed to have a makespan of~$3$
and for each possible case,
there exists a solution with makespan~$3$ for the entire sequence.
The claim follows because in all cases the makespan of~$A$ is at least~$4$.
\mbox{}\qed
\end{proof}

An obvious next step would be to try to match the general lower bound
 of~$\frac43$ by designing an algorithm that places each job on the
 most loaded machine where the bound of~$\frac43$ would not be
 violated.
However, even for $m=3$, this would not work, as seen by the input
 sequence $I=\langle \frac34, \frac14, \frac{5}{12}, \frac16,
 \frac{7}{12}, \frac56 \rangle$.
The algorithm would combine the first two jobs on one machine and the
 following two on another machine.
Since the optimal makespan at this point is~$\frac34$, the algorithm
will schedule the fifth job on the third machine.
When the last job arrives, all machines have a load of at least~$\frac{7}{12}$,
resulting in a makespan of $\frac{17}{12} > 1.4$.
Note that $I$ can be scheduled such that each machine has a load of
 exactly 1.
Since the algorithm has a makespan of 1 already after the second job,
 the online-bounded restriction is actually no restriction on \OPT for
 this sequence.

\section{Machine Scheduling: Santa Claus}
\label{sec:santa}
In contrast to makespan, the objective in Santa Claus scheduling
is to maximize the minimum load.
The problem is also known as machine covering.
Traditionally, the algorithm \GREEDY for this problem assigns any
new job to a machine having a minimum load in the schedule that
was created up to the time just before the job is added to the
solution (breaking ties arbitrarily). For identical machines, this
algorithm is equivalent to the greedy algorithm for makespan
minimization. Unlike the makespan minimization problem, where this
algorithm has online-bounded ratio of~$1$ only for two identical
machines, here we show that \GREEDY has an online-bounded ratio of
$1$ for any number of identical machines.

\begin{theorem}\label{idm}
For the Santa Claus problem on $m$ identical machines,
\GREEDY has online-bounded ratio~$1$.
\end{theorem}
\begin{proof}
Let a configuration be a multi-set of the current loads
on all of the machines, i.e., without any annotation of which machine is which.
As long as $\OPT_{\GREEDY}$ also assigns each job to a machine
with minimum load, the configurations of \GREEDY and $\OPT_{\GREEDY}$
are identical.

Consider the first time $\OPT_{\GREEDY}$ assigns a job $j$ to a
non-minimal load machine. If, when that job~$j$ arrives, there is
a unique machine with minimum load, $\OPT_{\GREEDY}$ would have a
worse objective value than \GREEDY after placing~$j$, so, by
definition of online-bounded analysis, this cannot happen. Now
consider the situation where $k\geq 2$ machines have minimum load.
Then, after processing~$j$, \GREEDY has $k-1$ machines with
minimum load, whereas $\OPT_{\GREEDY}$ has~$k$. In that case, no
more than $k-2$ further jobs can be given. This is seen as
follows: If $k-1$ jobs were given, \GREEDY would place one on each
of its $k-1$ machines with minimum load, and, thus, raise the
minimum. $\OPT_{\GREEDY}$, on the other hand, would not be able to
raise (at this step) the minimum of all of its $k$ machines with
minimum load, and would therefore not be at least as good as
\GREEDY; a contradiction.

Thus, $\OPT_{\GREEDY}$ can only have a different
configuration than \GREEDY after \GREEDY (and $\OPT_{\GREEDY}$) have obtained
their final (and identical) objective value,
and so, the online-bounded ratio of \GREEDY is~$1$.
\mbox{}\qed
\end{proof}

Next, we show that unlike the makespan minimization problem, for
which there is an online algorithm with online-bounded ratio of~$1$
for the case of two uniformly related machines
(Theorem~\ref{makespanrel}), such a result is impossible for
the Santa Claus problem.

\begin{theorem}
\label{SantaClausImpossibility}
For the Santa Claus problem on two uniformly related machines
with speed ratio~$s$,
no deterministic online algorithm has an online-bounded ratio
larger than~$\frac1s$.
\end{theorem}
\begin{proof}
For the setting of two uniformly related machines with speeds~$1$ and~$s$,
consider any online algorithm~$A$.
The input consists of exactly two jobs.  After
the first job is assigned by~$A$, the objective function value
remains zero, and only if the algorithm assigns the two jobs to
distinct machines, will it have a positive objective function
value. Thus, when there are only two jobs, $\OPT_A$ is simply the optimal
solution for the instance. The first job is of size~$1$. If $A$
assigns the job to the machine of speed~$s$, then the next job is
of size~$s$.  At this point $\OPT_A$ has value~$1$ (by assigning
the first job to the slower machine and the second to the faster
machine), but $A$ has either zero value (if both jobs are assigned
to the faster machine) or a value of~$\frac{1}{s}$.
In the second case where $A$ assigns the first job (of
size~$1$) to the slower machine of speed~$1$, the second job has
size~$\frac{1}{s}$.  At this point $\OPT_A$ has value
$\frac{1}{s}$ (by assigning the first job to the faster machine
and the second to the slower machine), but $A$ has either zero
value (if both jobs are assigned to the slower machine) or a value
of~$\frac{1/s}{s}$.
\mbox{}\qed
\end{proof}

Interestingly, the online-bounded ratio of \GREEDY matches this
bound, whereas post-\GREEDY does not. The algorithm post-\GREEDY
is relevant for uniformly related machines, and places a job on
the machine where its resulting completion time will be minimum,
which is not necessarily the machine with the smallest load when
the job arrives. \GREEDY also achieves the best possible
competitive ratio~$\frac 1{s+1}$~\cite{Eps05}.

\begin{theorem}\label{oneovers}
For the Santa Claus problem on two uniformly related machines
with speed ratio~$s$, the online-bounded ratio of~$\GREEDY$ is~$\frac1s$.
\end{theorem}
\begin{proof}
By Theorem~\ref{SantaClausImpossibility}, the online-bounded ratio
of~$\GREEDY$ is at most~$\frac{1}{s}$. Now we show that it is at
least~$\frac{1}{s}$. Assume $s>1$ (otherwise the result follows
from Theorem~\ref{idm}). For a given input, $I$, and the output
of~$\GREEDY$ for this input, let $j$ denote a job of maximum
completion time. Let $x$ denote the load of the machine with
job~$j$ just before $j$ is assigned. Let $y$ denote the load of
the other machine at the same time. By the definition
of~$\GREEDY$, $y \geq x$. Let $y+z$ denote the final load of the
machine whose previous load was $y$ (the machine that does not
receive~$j$). The value $y+z$ is also the value of~$\GREEDY$ on
this input.

Consider a solution by $\OPT_{\GREEDY}$ (that is, an online-bounded
optimal solution). Let $t_1$ and~$t_2$ denote the loads of the
machines of speeds~$1$ and~$s$, respectively, before~$j$ is
assigned. By the definition of such an optimal solution,
$\min\{t_1,t_2\} \geq x$. We split the analysis into two cases,
based on which machine receives~$j$ in the output of~$\GREEDY$.

Assume that the machine of speed~$1$ runs~$j$ in the schedule
of~$\GREEDY$. Just before $j$ arrives, the total size of jobs
is~$x+sy$. We find $t_1 = x+sy-st_2 \leq  x+sy-sx<sy$ (since
$s>1$) and $t_2 =\frac{x+sy-t_1}{s} \leq y$. The total size of
jobs arriving strictly after $j$ is~$sz$, and in the optimal
solution, the load of the machine that does not receive~$j$ is at
most $\max\{t_1+sz,t_2+z\} \leq \max\{sy+sz,y+z\}=s(y+z)$. Thus,
$\OPT_{\GREEDY}(I) \leq s(y+z) \leq s \cdot \GREEDY(I)$.

Next, assume that the machine of speed~$s$ runs~$j$ in the
schedule of~$\GREEDY$. Just before $j$ arrives, the total size of
jobs is~$sx+y$. We find $t_1 = sx+y-st_2 \leq  y$ and $t_2
=\frac{sx+y-t_1}{s} \leq \frac{sx+y-x}{s}\leq \frac{(s-1)y+y}{s} =
y$ (by~$x \leq y$). The total size of jobs arriving strictly
after~$j$ is~$z$, and in the optimal solution, the load of the
machine that does not receive~$j$ is at most
$\max\{t_1+z,t_2+\frac zs\} \leq \max\{y+z,y+\frac zs\}=y+z$, and
in this case the solution of~$\GREEDY$ is optimal. \mbox{}\qed
\end{proof}

\section{Classic Bin Packing and Bin Covering}
\label{sec:bin}
In classic bin packing, the input is a sequence of items of sizes
 $s$, $0 < s \leq 1$, that should be packed in as few bins of size 1
 as possible. We say that a bin is open if at least one item has been
placed in the bin.
An Any-Fit algorithm is an algorithm that never opens a new bin
if the current item fits in a bin that is already open.

\begin{theorem}\label{atleast32}
Any Any-Fit algorithm has an online-bounded ratio of at least~$\frac32$.
\end{theorem}
\begin{proof}
The adversary sequence,~$I$, consists of three parts, $I_1$, $I_2$,
 and~$I_3$.
$I_1$ and~$I_2$ contain just a few items each, while $I_3$
 contains $4(n-1)$ items, for some large integer~$n$.
We show that any Any-Fit algorithm, \ALG,
uses $3(n-1)$ bins for $I_3$, whereas \OPTALG
uses only $2(n-1)$ bins.

The first part of~$I$ consists of three items:
 $$I_1 = \left\langle \frac{2}{3}, \frac{5}{12}, \frac{1}{4}
\right\rangle$$
Any algorithm will have to use two bins for the first two items, and
 any Any-Fit algorithm will pack the third item in one of these two
 bins.
The second part of the sequence depends on whether \ALG packs the third item
 in the first or the second bin.

If \ALG packs the item of size $\frac{1}{4}$ in the first bin together
 with the item of size $\frac{2}{3}$, the second part of the sequence
 contains four items:
 $$I_2 = \left\langle \frac{1}{3}, \frac{1}{3}, \frac{1}{2}-n\eps,
         \frac{1}{2}+n\eps \right\rangle,$$
 for some small \eps, $0<\eps<\frac{1}{12n}$.
Since \ALG is an Any-Fit algorithm, it packs the first of these four items
 in the second bin and then opens a third bin for the next two items
 and a fourth bin for the last item.
\OPTALG uses only three bins in total for $I_1$ and~$I_2$, combining the
 items of sizes $\frac{2}{3}$ and~$\frac{1}{3}$ in the first bin and
 the items of sizes $\frac{1}{2}-n\eps$ and~$\frac{1}{2}+n\eps$ in the
 third bin.

If \ALG packs the item of size $\frac{1}{4}$ in the second bin
 together with the item of size~$\frac{5}{12}$, the second part of
 the sequence contains only one item:
 $$I_2 = \left\langle \frac{7}{12} \right\rangle$$
\ALG will have to open a new bin for this item.
\OPTALG, on the other hand, will combine the items of sizes~$\frac{2}{3}$
 and~$\frac{1}{4}$ in one bin and the items of size~$\frac{5}{12}$
 and~$\frac{7}{12}$ in another bin.

In both cases, \ALG has now opened one more bin than \OPTALG and each of
 \ALG's bins is filled to at least~$\frac{1}{2}+n\eps$.

The last part of the sequence consists of $n-1$ consecutive subsequences:
 $$\left\langle \frac{1}{2}-i\eps, \frac{1}{2}-i\eps,
\frac{1}{2}+i\eps, \frac{1}{2}+i\eps \right\rangle, i=n-1,n-2,\ldots,1$$
For each of these $n-1$ subsequences,~$I'$, none of the four items fit in
 any of the bins opened before the arrival of the first item of~$I'$.
Hence, \ALG uses $3$~bins for each of the $n-1$ subsequences of
 $I_3$, $3(n-1)$ bins in total.
\OPTALG, on the other hand, will put the first two items of each
 subsequence in separate bins and pack the last two items in the same
 two bins.
This is allowed, since \ALG has opened one more bin than \OPTALG, already
 before the arrival of the first item of~$I_3$.
In this way, \OPTALG uses only $2(n-1)$ bins for~$I_3$.
Since both algorithms use only a constant number of bins for~$I_1$ and~$I_2$,
the ratio of \ALG bins to \OPTALG bins tends to~$3/2$ as $n$
tends to infinity.
\mbox{}\qed
\end{proof}

In classic bin covering, the input is as in bin packing,
and the goal is to assign items to bins so as to maximize the number of bins
whose total assigned size is at least~$1$.
For this problem, it is known that a simple greedy algorithm
(which assigns all items to the active bin until the total size assigned
to it becomes~$1$ or larger, and then it moves to the next bin
and defines it as active) has the best possible competitive ratio~$\frac 12$.
The negative result~\cite{CT88} is proven using inputs
where the first batch of items consists of a large number of very small items,
and it is followed by a set of large identical items of sizes close to~$1$
(where the exact size is selected based on the actions of the algorithm).
The total size of the very small items is strictly below~$1$,
so as long as large items were not presented yet, the value of any algorithm
is zero. An optimal offline solution packs the very small items
such that packing every large item results in a bin whose contents
have a total size of exactly~$1$. Thus, no algorithm can perform better
on any prefix, and this construction shows that the online-bounded ratio
is at most~$\frac 12$.

\section{Dual Bin Packing}
\label{sec:dual}
Dual bin packing is like the classic bin packing
problem, except that there is only a limited number,~$n$, of bins and
the goal is to pack as many items in these $n$ bins as possible.
Known results concerning the competitive ratio on accommodating sequences
can be used to obtain results for the online-bounded ratio.

\paragraph{Online-Bounded Ratio vs.\ Competitive Ratio on Accommodating Sequences.}
In general, accommodating sequences~\cite{BL99,BLN01} are defined to
be those sequences for which \OPT does not get a better result by
having more resources.
For the dual bin packing problem, accommodating sequences are
 sequences of items that can be fully accommodated in the $n$~bins,
 i.e., \OPT packs all items.

We show that, for a large class of algorithms for dual bin packing
containing First-Fit and Best-Fit, the online-bounded ratio is the
 same as the competitive ratio on accommodating sequences.
To show that this does not hold for all algorithms, we also give an
example of a~$\frac23$-competitive algorithm
 on accommodating sequences that has an online-bounded ratio of~$0$.

Dual bin packing is an example of a problem in a larger class of
problems which includes the seat reservation problem discussed
below. A problem is an \emph{accept/reject accommodating problem}
if algorithms can only accept or reject requests (and they act on
accepted requests only), the goal is to accept as many requests as
possible, and the accommodating sequences are those where \OPT
accepts all requests.

\begin{theorem}
\label{accSeq}
For any online algorithm \ALG for any accept/reject accommodating problem, the
competitive ratio of \ALG on accommodating sequences is equal to
the online-bounded ratio of \ALG on accommodating sequences.
\end{theorem}
\begin{proof}
For any accommodating sequence, \OPT rejects no items.
Thus, the requirement that at any point in time, \OPT has packed at
 least as many items as \ALG does not change the behavior of \OPT.
This means that, for accommodating sequences, the
competitive ratio and the online-bounded ratio are identical.
\mbox{}\qed
\end{proof}

Note that this result applies to all algorithms for dual bin packing.
Since any accommodating sequence is also a valid adversarial sequence
for the case with no restrictions on the sequences, we obtain the
following corollary of Theorem~\ref{accSeq}.

\begin{corollary}
\label{dualBPupper}
For any online algorithm \ALG for any accept/reject accommodating problem,
any upper bound on the
 competitive ratio of \ALG on accommodating sequences is also an upper
 bound on the online-bounded ratio of \ALG.
\end{corollary}

A {\em fair} algorithm for dual bin packing is an algorithm that never
 rejects an item that it could fit in a bin.
A {\em rejection-invariant} algorithm is an
 algorithm that does not change its behavior based on rejected items.

\begin{theorem}
\label{dualBPlower}
For any fair, rejection-invariant dual bin packing algorithm \ALG, the
 online-bounded ratio of \ALG equals the competitive ratio of \ALG on
 accommodating sequences.
\end{theorem}
\begin{proof}
The upper bound follows from Corollary~\ref{dualBPupper}. For the
lower bound, we show that, for any input sequence~$I$, there is an
accommodating sequence~$I'$ such that
$\OPT_{\ALG}(I)=\OPT_{\ALG}(I')=\OPT(I')$ and $\ALG(I)=\ALG(I')$.

Assume without loss of generality that for $I$ it holds that any
request is accepted either by $\ALG$ or by $\OPT_{\ALG}$ (or by
both of them). This can be assumed as \ALG is rejection-invariant,
and all requests rejected by both \ALG and $\OPT_{\ALG}$ can be
removed from $I$ without changing the action of \ALG, and thus the
action of $\OPT_{\ALG}$ is unchanged as well.

Let $\langle o_1, o_2, \ldots, o_k \rangle$ be the subsequence of~$I$
 consisting of the items that are packed by $\OPT_{\ALG}$ but not by \ALG.
Similarly, let $\langle a_1, a_2, \ldots, a_l \rangle$ be the
 subsequence of~$I$ consisting of the items packed by \ALG but not by
$\OPT_{\ALG}$. Clearly, $k \geq l$. Furthermore, for each $i$, $1
\leq i \leq l$, $o_i$ arrives before
 $a_i$, since at any point in time, $\OPT_{\ALG}$ must have packed at least as
 many items as \ALG.
This means that $o_i$ is larger than $a_i$, since \ALG is fair and
 packs $a_i$ after rejecting~$o_i$.

Let $I'$ be the subsequence of $I$ resulting from the removal of
$\langle o_1, o_2, \ldots, o_l \rangle$. We claim that $I'$ is
accommodating. The packing of $\OPT_{\ALG}$ consists of all items
of $I$ excluding $\langle a_1, a_2, \ldots, a_l \rangle$. This
packing is adapted for $I'$ as follows. For $i=1,2,\ldots,l$,
$o_i$ is replaced with $a_i$ (which is smaller). In this packing,
all items of $I'$ are packed. Therefore, $OPT(I')=\OPT_{\ALG}(I)$
(since in the packing we just defined, the number of packed items is
unchanged, and all items of $I'$ are packed). Recall that since
$I'$ is an accommodating sequence, it holds that $\OPT_{\ALG}(I')=\OPT(I')$.
Since \ALG rejects all items of $I \setminus I'$ and is
rejection-invariant, its behavior for $I$ and $I'$ is the same and
$\ALG(I')=\ALG(I)$. \mbox{}\qed
\end{proof}

One algorithm which is fair and rejection-invariant is First-Fit,
which packs each item in the first bin it fits in (and rejects it if
no such bin exists).  Another example of a fair, rejection-invariant
algorithm is Best-Fit, which packs each item in a most full bin that
can accommodate it. Worst-Fit is the algorithm that packs each item in
a most empty bin.

\begin{corollary}\label{BF58}
Both Best-Fit and First-Fit have online-bounded ratios of~$\frac58$.
Worst-Fit has an online-bounded ratio of~$\frac12$.
\end{corollary}
\begin{proof}
For Best-Fit, the corollary follows from Theorem~\ref{dualBPlower},
since Corollary~2 and Theorem~4 in~\cite{BFLN03} imply that
Best-Fit's competitive ratio on accommodating sequences is~$\frac58$.
Theorem~4 in~\cite{BFLN03} does not specifically mention Best-Fit, but
it is well known that any negative result for the competitive ratio
of First-Fit holds for any Any-Fit algorithm.
This is because, for any sequence~$I$, an adversary can force any
Any-Fit algorithm to produce the First-Fit packing of~$I$ by
permuting~$I$ such that the items packed by First-Fit are given
first, in the order they appear in the bins in the First-Fit packing.

For First-Fit, the corollary follows from Theorem~\ref{dualBPlower},
since Corollary~1 and Theorem~4 in~\cite{BFLN03} imply that
First-Fit's competitive ratio on accommodating sequences is~$\frac58$.

For Worst-Fit, the corollary follows from Theorem~\ref{dualBPlower}, since
Theorems~1 and~5 in~\cite{BFLN03} imply that
Worst-Fit's competitive ratio on accommodating sequences is~$\frac12$.
\mbox{}\qed
\end{proof}

For the first part of the corollary below, note that the fairness restriction gives
rise to a lower bound, i.e., a guarantee of at least a certain
competitive ratio.
\begin{corollary}\label{cor12}
Any fair, rejection-invariant dual bin packing algorithm has an
on\-line-bound\-ed ratio of at least~$\frac12$.
Any dual bin packing algorithm has an
online-bounded ratio of at most~$\frac67$.
\end{corollary}
\begin{proof}
The first part follows from Theorem~\ref{dualBPlower} above combined with
Theorem~1 in~\cite{BFLN03}.
The second part follows from Corollary~\ref{dualBPupper} above combined with
 Theorem~3 in~\cite{BFLN03}.
\mbox{}\qed
\end{proof}

The algorithm Unfair-First-Fit (\UFF) defined in~\cite{ABEFLN02} is
 designed to work well on accommodating sequences.
Whenever an item larger than $\frac12$ arrives, \UFF rejects the item
 unless it will bring the number of accepted items below~$\frac23$ of
 the total number of items that are accepted by an optimal solution
of the prefix of items given so far (for an accommodating sequence
this is the number of items in the prefix). Accepted items are
packed using First-Fit. The competitive ratio of \UFF on
accommodating sequences is~$\frac23$~\cite{ABEFLN02}. We show
that, in contrast to Theorem~\ref{dualBPlower}, \UFF has an
 online-bounded ratio of~$0$.

\begin{theorem}\label{UFF0}
Unfair-First-Fit has an online-bounded ratio of~$0$.
\end{theorem}
\begin{proof}
For a fixed $n \geq 2$, consider the following input sequence $I$,
for some small $\eps>0$, where $\eps=\frac{1}{2N}$ for some large
integer $N$.  The sequence starts with $\langle \eps, \eps,
\frac12+\eps, 1-2\eps \rangle$, followed by $\langle
\frac12\rangle$ repeated $2(n-1)$ times, and $\langle \eps
\rangle$ repeated $N-3$ times. The first four items can be packed
into two bins. Unfair-First-Fit rejects the third item and accepts
the other three items among the first four items. Moreover,
Unfair-First-Fit packs those three items into one bin. Next,
Unfair-First-Fit accepts and packs all items of size $\frac 12$
into its $n-1$ remaining bins. It is forced to reject all
remaining items (each of which has size $\eps$). For
$\OPT_{\UFF}$, it is possible to reject the fourth item instead of
the third one, and as a result, it can pack all other items (it
also packs the items of sizes $\frac 12$ in pairs into bins of
indices $2,3,\ldots,n$). After all items of size $\frac 12$ have
been presented, $\OPT_{\UFF}$ has empty space, which is filled by
small items, resulting in the ratio
$\frac{\UFF(I)}{\OPT_{\UFF}(I)}=\frac{2n+1}{2n+N-2}$, tending to
$0$ as $N$ grows to inifinity.
\mbox{}\qed
\end{proof}

\section{Unit Price Seat Reservation}
\label{sec:unit}
In the seat reservation problem, there is a train with $n$ seats
traveling from station $1$ to station $k$.
The input is a sequence of requests for getting a seat from a station
$i$ to a station $j>i$.
Two requests can be assigned the same seat, if the end station of one
request is no larger than the start station of the other request.
Algorithms for the problem are required to be fair, i.e., a request
cannot be rejected if at least one seat can accommodate the request.
In the unit price version, the objective is to maximize the number of
accepted requests (requests that are assigned a seat), and in the
proportional price version, the objective is to maximize the total
length (end station minus start station) of the accepted requests.

Since both versions of the seat reservation problem
have competitive ratios $\Theta(1/k)$, the problem
has often been studied using the \emph{competitive ratio on accommodating
sequences}, which for the seat reservation problem
restricts the input sequences considered to those
where $\OPT$ could have accepted all of the requests.
For proportional price seat reservation, the competitive ratio remains
$\Theta(1/k)$, even on accommodating sequences.
However, for the unit price version, the optimal competitive ratio on
accommodating sequences is $\frac12$~\cite{BL99,BBEFJLLS03}.
Thus, by Theorem~\ref{accSeq}, the optimal online-bounded ratio on
accommodating sequences is $\frac12$.

On general sequences, the optimal online-bounded ratio for unit price
seat reservation is essentially as bad as the competitive ratio.
This is true, even though both the original proof,
showing that no deterministic
online algorithm is more than
$\frac{8}{k+5}$-competitive~\cite{BL99}, and the proof improving this
to $\frac{4}{k-2\sqrt{k-1}+4}$~\cite{MO10}, used an optimal
offline algorithm which rejected some requests before
the online algorithm did. The main idea in these proofs was that
the adversary could
give small request intervals which \OPT could place differently from
the algorithm, allowing it to reject some long intervals and still
be fair. Rejecting long
intervals allowed it to accept many short intervals which the algorithm
was forced to reject.
By using small intervals involving only the last few stations,
one can force the online algorithm to reject intervals
early. Then, giving nearly the same sequence as for the $\frac{8}{k+5}$
bound, using two fewer stations,
\OPT can still reject the same long intervals and do just as badly
asymptotically.
Note that in the proof, the $[k-3,k-2)$ intervals are used both
in the initial part, targeting the last few stations, and in the main
construction that follows.

\begin{theorem}
\label{seatreservation}
No deterministic online algorithm for the unit price seat reservation
problem has an online-bounded ratio of more than~$\frac{11}{k+7}$.
\end{theorem}
\begin{proof}
Assume the number of seats~$n$ is divisible by~$4$.
We compare an online algorithm~$A$ to
an algorithm~$O$, which follows all of the rules for~$\OPT_A$.
The adversary gives $n/2$ pairs of requests for $[k-3,k-2)$ and $[k-1,k)$.
Suppose the online algorithm,~$A$, places these intervals
such that after these requests
there are exactly $r$ seats which contain two intervals.
One of two cases will occur:
\begin{itemize}
\item Case~1: $r\geq n/4$, or
\item Case~2: $r<n/4$.
\end{itemize}
If Case~1 occurs, the algorithm~$O$ places all of the first $n$ intervals on
separate seats. Next there will be $n/2$ requests to $[k-2,k)$,
followed by $n/2$ requests to
$[k-3,k-1)$. $A$ will reject $r\geq n/4$ of these requests to
$[k-3,k-1)$, but $O$
accepts them all. This allows $O$ to later reject $n/4$ intervals
that $A$ accepts.

If Case~2 occurs, $O$ pairs up the first $n$ intervals, placing
two on each of the first $n/2$ seats.
Next there will be $n/2$ requests
for $[k-3,k)$ intervals, and $A$ rejects $n/2-r\geq n/4$ of them, but $O$
accepts them. This also allows $O$ to later reject $n/4$ intervals
that $A$ accepts.

Note that in both cases, $A$ now has at least $n/4$ seats with
the interval $[k-3,k-2)$ free, while $O$ has none.
Let $a$ be the number of intervals accepted by $A$ up to
this point and $o$ be the number of intervals accepted by~$O$. The
value~$a$ is at most $7n/4$ in Case~1 and at most $5n/4$ in Case~2.
The value of $o$ is $2n$ in Case~1 and $3n/2$ in Case~2.

Now there will be $n/4$ requests for $[1,k-2)$, that $A$ accepts
and $O$ rejects. These are followed by $3n/4$ requests for
$[1,k-3)$ that both $A$ and $O$ accept. Finally, there are
$n/4$ requests for each of the intervals $[i,i+1)$, $1\leq i\leq k-4$.
$A$ cannot accept any of these $n(k-4)/4$ requests, but $O$
accepts all of them.

For this adversarial sequence,~$I_A$,
\begin{eqnarray*}
\frac{A(I_A)}{\OPT_A(I_A)} & \leq & \frac{A(I_A)}{O(I_A)} \\
& \leq & \frac{a+n}{o+3n/4+n(k-4)/4} \\
& \leq &  \frac{7n/4+n}{2n+3n/4+n(k-4)/4} \\
& = &  \frac{11}{k+7} \ ,
\end{eqnarray*}
where the third inequality holds because in both cases
we have $o \geq a+\frac n4$ and $a\leq \frac{7n}{4}$,
so the bounds in Case~1 give the larger result.
\mbox{}\qed
\end{proof}

Using a similar proof, one can show that the online-bounded ratios
of First-Fit and Best-Fit are at most~$\frac{5}{k+1}$. The major
difference is that in the first part, First-Fit and Best-Fit each
reject $n/2$ intervals, so in the second part, $O$ can also reject
$n/2$ intervals.
Since any online algorithm for the unit price problem is $2/k$-competitive,
any online algorithm for the unit price problem has an
online-bounded ratio of at least~$2/k$.

\bibliography{refs}

\begin{thebibliography}{10}
\providecommand{\url}[1]{{#1}}
\providecommand{\urlprefix}{URL }
\expandafter\ifx\csname urlstyle\endcsname\relax
  \providecommand{\doi}[1]{DOI~\discretionary{}{}{}#1}\else
  \providecommand{\doi}{DOI~\discretionary{}{}{}\begingroup
  \urlstyle{rm}\Url}\fi

\bibitem{A97}
Albers, S.: On the influence of lookahead in competitive paging algorithms.
\newblock Algorithmica \textbf{18}, 283--305 (1997)

\bibitem{AFG02}
Albers, S., Favrholdt, L.M., Giel, O.: On paging with locality of reference.
\newblock J. Comput. Syst. Sci. \textbf{70}(2), 145--175 (2005)

\bibitem{ADLO07}
Angelopoulos, S., Dorrigiv, R., L{\'{o}}pez-Ortiz, A.: On the separation and
  equivalence of paging strategies.
\newblock In: 18th ACM-SIAM Symposium on Discrete Algorithms (SODA), pp.
  229--237 (2007)

\bibitem{ABEFLN02}
Azar, Y., Boyar, J., Epstein, L., Favrholdt, L.M., Larsen, K.S., Nielsen, M.N.:
  Fair versus unrestricted bin packing.
\newblock Algorithmica \textbf{34}(2), 181--196 (2002)

\bibitem{AE98}
Azar, Y., Epstein, L.: On-line machine covering.
\newblock Journal of Scheduling \textbf{1}(2), 67--77 (1998)

\bibitem{AR01}
Azar, Y., Regev, O.: On-line bin-stretching.
\newblock Theoretical Computer Science \textbf{268}(1), 17--41 (2001)

\bibitem{BBEFJLLS03}
Bach, E., Boyar, J., Epstein, L., Favrholdt, L.M., Jiang, T., Larsen, K.S.,
  Lin, G., van Stee, R.: Tight bounds on the competitive ratio on accommodating
  sequences for the seat reservation problem.
\newblock Journal of Scheduling \textbf{6}(2), 131--147 (2003)

\bibitem{BS06}
Bansal, N., Sviridenko, M.: The {Santa Claus} problem.
\newblock In: 38th Annual ACM Symposium on the Theory of Computing (STOC), pp.
  31--40 (2006)

\bibitem{BDB94}
Ben-David, S., Borodin, A.: A new measure for the study of on-line algorithms.
\newblock Algorithmica \textbf{11}(1), 73--91 (1994)

\bibitem{BIRS95}
Borodin, A., Irani, S., Raghavan, P., Schieber, B.: Competitive paging with
  locality of reference.
\newblock Journal of Computer and System Sciences \textbf{50}(2), 244--258
  (1995)

\bibitem{BFKM15}
Boyar, J., Favrholdt, L., Mikkelsen, J., Kudahl, C.: Advice complexity for a
  class of online problems.
\newblock In: 32nd International Symposium on Theoretical Aspects of Computer
  Science, ({STACS}), \emph{Leibniz International Proceedings in Informatics},
  vol.~30, pp. 116--129 (2015)

\bibitem{BF07}
Boyar, J., Favrholdt, L.M.: The relative worst order ratio for on-line
  algorithms.
\newblock ACM Transactions on Algorithms \textbf{3}(2), article 22, 24 pages
  (2007)

\bibitem{BFLN03}
Boyar, J., Favrholdt, L.M., Larsen, K.S., Nielsen, M.N.: {Extending the
  Accommodating Function}.
\newblock Acta Informatica \textbf{40}(1), 3--35 (2003)

\bibitem{BL99}
Boyar, J., Larsen, K.: The seat reservation problem.
\newblock Algorithmica \textbf{25}, 403--417 (1999)

\bibitem{BLN01}
Boyar, J., Larsen, K.S., Nielsen, M.N.: The accommodating function---a
  generalization of the competitive ratio.
\newblock SIAM Journal on Computing \textbf{31}(1), 233--258 (2001)

\bibitem{B98}
Breslauer, D.: On competitive on-line paging with lookahead.
\newblock Theoretical Computer Science \textbf{209}(1--2), 365--375 (1998)

\bibitem{CLLLT11}
Chan, S.H., Lam, T.W., Lee, L.K., Liu, C.M., Ting, H.F.: Sleep management on
  multiple machines for energy and flow time.
\newblock In: L.~Aceto, M.~Henzinger, J.~Sgall (eds.) Automata, Languages and
  Programming (ICALP), \emph{LNCS}, vol. 6755, pp. 219--231. Springer-Verlag
  Berlin Heidelberg (2011)

\bibitem{CS80}
Cho, Y., Sahni, S.: Bounds for list schedules on uniform processors.
\newblock {SIAM} Journal on Computing \textbf{9}(1), 91--103 (1980)

\bibitem{CT88}
Csirik, J., Totik, V.: On-line algorithms for a dual version of bin packing.
\newblock Discrete Applied Mathematics \textbf{21}, 163--167 (1988)

\bibitem{DLM09}
Dorrigiv, R., L{\'{o}}pez-Ortiz, A., Munro, J.I.: On the relative dominance of
  paging algorithms.
\newblock Theoretical Computer Science \textbf{410}, 3694--3701 (2009)

\bibitem{EKL13}
Ehmsen, M.R., Kohrt, J.S., Larsen, K.S.: {List Factoring and Relative Worst
  Order Analysis}.
\newblock Algorithmica \textbf{66}(2), 287--309 (2013)

\bibitem{Eps05}
Epstein, L.: Tight bounds for bandwidth allocation on two links.
\newblock Discrete Applied Mathematics \textbf{148}(2), 181--188 (2005)

\bibitem{EFK06}
Epstein, L., Favrholdt, L.M., Kohrt, J.S.: Separating online scheduling
  algorithms with the relative worst order ratio.
\newblock Journal of Combinatorial Optimization \textbf{12}(4), 363--386 (2006)

\bibitem{ENSSW01}
Epstein, L., Noga, J., Seiden, S.S., Sgall, J., Woeginger, G.J.: Randomized
  online scheduling on two uniform machines.
\newblock Journal of Scheduling \textbf{4}(2), 71--92 (2001)

\bibitem{Giannakopoulos15}
Giannakopoulos, Y., Koutsoupias, E.: Competitive analysis of maintaining
  frequent items of a stream.
\newblock Theoretical Computer Science \textbf{562}, 23--32 (2015)

\bibitem{G66}
Graham, R.L.: Bounds for certain multiprocessing anomalies.
\newblock Bell Systems Technical Journal \textbf{45}, 1563--1581 (1966)

\bibitem{KP00}
Kalyanasundaram, B., Pruhs, K.: Speed is as powerful as clairvoyance.
\newblock Journal of the ACM \textbf{47}(4), 617--643 (2000)

\bibitem{KMRS88}
Karlin, A.R., Manasse, M.S., Rudolph, L., Sleator, D.D.: Competitive snoopy
  caching.
\newblock Algorithmica \textbf{3}, 79--119 (1988)

\bibitem{KPR00}
Karlin, A.R., Phillips, S.J., Raghavan, P.: Markov paging.
\newblock SIAM Journal on Computing \textbf{30}(3), 906--922 (2000)

\bibitem{K96}
Kenyon, C.: Best-fit bin-packing with random order.
\newblock In: 7th ACM-SIAM Symposium on Discrete Algorithms (SODA), pp.
  359--364 (1996)

\bibitem{KPn00}
Koutsoupias, E., Papadimitriou, C.H.: Beyond competitive analysis.
\newblock SIAM Journal on Computing \textbf{30}(1), 300--317 (2000)

\bibitem{MO10}
Miyazaki, S., Okamoto, K.: Improving the competitive ratios of the seat
  reservation problem.
\newblock In: 6th {IFIP} {TC} 1/WG 2.2 International Conference on Theoretical
  Computer Science {(IFIP TCS)}, \emph{{IFIP} Advances in Information and
  Communication Technology}, vol. 323, pp. 328--339. Springer (2010)

\bibitem{R91}
Raghavan, P.: A statistical adversary for on-line algorithms.
\newblock In: On-Line Algorithms, \emph{Series in Discrete Mathematics and
  Theoretical Computer Science}, vol.~7, pp. 79--83. American Mathematical
  Society (1992)

\bibitem{ST85}
Sleator, D.D., Tarjan, R.E.: Amortized efficiency of list update and paging
  rules.
\newblock Communications of the ACM \textbf{28}(2), 202--208 (1985)

\bibitem{Woe97}
Woeginger, G.J.: A polynomial-time approximation scheme for maximizing the
  minimum machine completion time.
\newblock Operations Research Letters \textbf{20}(4), 149--154 (1997)

\bibitem{Y91}
Young, N.: Competitive paging and dual-guided algorithms for weighted caching
  and matching (thesis).
\newblock Tech. Rep. CS-TR-348-91, Computer Science Department, Princeton
  University (1991)

\bibitem{Y94}
Young, N.E.: The $k$-server dual and loose competitiveness for paging.
\newblock Algorithmica \textbf{11}, 525--541 (1994)

\end{thebibliography}
\bibliographystyle{spmpsci}

\end{document}